\documentclass{article}
\usepackage{cite,comment}
\usepackage{amsmath,amssymb,amsfonts}
\usepackage{algorithmic}
\usepackage{graphicx}
\usepackage{textcomp}
\usepackage{xcolor}
\usepackage{bbm}
\usepackage{amsthm,amsmath,fullpage,setspace,amsfonts,amsthm,placeins}
\usepackage{graphicx}
\usepackage{mathrsfs}
\usepackage{amsmath}
\usepackage{xcolor}
\usepackage{arydshln}
\usepackage{subfig}
\usepackage{tikz,comment}
\usetikzlibrary{decorations.pathreplacing}

\newtheorem{thm}{Theorem}
\newtheorem{prop}{Proposition}
\newtheorem{lemma}{Lemma}
\newtheorem{cor}{Corollary}
\newtheorem{definition}{Definition}

\newcommand{\Z}{\mathbb{Z}}
\newcommand{\R}{{R}_{\cal S}}
\newcommand{\Rg}{{\cal R}_{\cal S}}

\newcommand{\C}{\cal C}
\newcommand{\FRp}{(n,\alpha,\rho)}
\newcommand{\abbr}{LBFR } 
\newcommand{\abbrnospace}{LBFR}
\newcommand{\mkw}[1]{\textbf{ \textcolor{blue}{[ #1 --mary ]} } }

\newcommand{\V}{{\cal V}}
\newcommand{\cS}{{\cal S}}
\newcommand{\poS}{\partial_{out}(S)}
\newcommand{\Gp}{G'_{\C}}
\newcommand{\bF}{\mathbb{F}}

\begin{document}



\title{Load-Balanced Fractional Repetition Codes\thanks{A preliminary version of this work appeared at ISIT 2018.} }
%
%
%
\author{Alexandra Porter, Shashwat Silas, and Mary Wootters\thanks{AP is partially supported by the National Science Foundation Graduate Research Fellowship under Grant No. DGE-1656518.  SS is partially supported by a Google Graduate Fellowship in Computer Science. SS, AP and MW are partially supported by NSF Grants CCF-1657049, CCF-1814629 and CCF-1844628. }}

\markboth{Journal }%
{Porter, Silas and Wootters: Embedded Index Coding}
%



\maketitle

\begin{abstract}
We introduce \em load-balanced fractional repetition \em (\abbrnospace) codes, which are a strengthening of fractional repetition (FR) codes.  \abbr codes have the additional property that multiple node failures can be sequentially repaired by downloading no more than one block from any other node. This allows for better use of the network, and can additionally reduce the number of disk reads necessary to repair multiple nodes. 
We characterize \abbr codes in terms of their adjacency graphs, and use this characterization to present explicit constructions of \abbr codes with storage capacity comparable to existing FR codes.  Surprisingly,
in some parameter regimes, our constructions of \abbr codes match the parameters of the best constructions of FR codes.
\end{abstract}


\section{Introduction}

We would like to store a file in a distributed manner across n storage nodes (servers) with the following objective: in the event of the failure of some nodes, the system can be efficiently repaired without data loss. In this paper, we focus on two notions of ``efficient:" \em network bandwidth \em and \em disk access. \em  That is, when a node (or a collection of nodes) fails, we would like to minimize the amount of data that is read and sent by surviving nodes in order to repair the system. This problem has been well-studied in the area of regenerating codes, e.g.~\cite{dimakis2010network,rashmi2011optimal}.  Fractional repetition (FR) codes are a special case of regenerating codes which have many desirable properties. In this work, we present \em load-balanced fractional repetition \em (\abbrnospace) codes, which are a strengthening of FR codes.

Fractional repetition codes have been well-studied since they were introduced in 2010; see~\cite{el2010fractional, zhu2014low,park2016construction,zhu2014replication,yu2014irregular,zhu2015heterogeneity,silberstein2015optimal}. FR codes use the minimum possible network bandwidth and disk access (the amount of data read by the surviving nodes in order to be sent) to repair node failures.  They have a computationally inexpensive repair process known as uncoded or exact repair which is essentially data copying from surviving nodes. Finally, there are constructions of FR codes which have simple combinatorial descriptions. The literature on FR codes has focused on the case of repairing a single node failure, while for practical and mathematical considerations, it is also interesting to think about regenerating codes in the case of multiple node failures. Indeed, multiple failures have been studied in cooperative regenerating codes~\cite{shum2013cooperative} and lazy repair schemes for regenerating codes~\cite{kermarrec2011repairing} in the past. We show that there are FR codes which can repair multiple node failures in a load balanced manner using a simple algorithm.

LBFR codes have a simple combinatorial characterization as a subset of FR codes and they have all the same benefits of FR codes for the repair of a single node. Additionally, they can handle multiple failures while load balancing the work of the helper nodes involved in the repair process. Indeed, the repair of multiple failures uses a particularly simple greedy algorithm and perhaps surprisingly we even find that in some parameter settings there are LBFR codes with the same storage capacity as optimal FR codes. These codes make better use of the network during repair, and as discussed below can improve the number of disk accesses required during the repair of multiple nodes and the time it takes to do repair in certain communication models.

\subsection{Fractional repetition codes} \em Fractional repetition \em codes were introduced in~\cite{el2010fractional} as a generalization of the codes in~\cite{shah2012distributed}.  
FR codes work by replicating data.  
A single data block $B$ is replicated $\rho$ times and each copy is stored on a different node; each node stores $\alpha$ blocks.  When a node fails it is repaired using copies of each of its $\alpha$ blocks stored elsewhere on the system. 
Formally, an FR code is defined as follows:
\begin{definition}[~\cite{el2010fractional}]\label{def:fr}
A $(n,\alpha,\rho)$ fractional repetition code is a collection of $n$ subsets $N_1,...,N_n$ of $[\theta]:=\{1,2,...,\theta\}$ with $n\alpha = \rho \theta$ such that:
\begin{enumerate}
\item $|N_i| = \alpha$ for $i = 1,...,n$;
\item For each $b \in [\theta]$, $|\{N_i: b \in N_i\}| = \rho$; $\rho$ is called the repetition degree.
\end{enumerate}
\end{definition}
Because each block is replicated $\rho$ times, a $(n, \alpha, \rho)$ FR code can tolerate up to $\rho - 1$ node failures.
Additionally, most constructions of FR codes naturally have a \em load-balancing \em property 
which ensures that a single failed node can be repaired by downloading a single block from each of $\alpha$ surviving nodes (as opposed to $\alpha$ blocks from a single node).  In this work, we'll extend this latter property to $\rho-1$ failures.

FR codes can be used as a building block to obtain codes---called \em DRESS codes\em\cite{el2010fractional}---with the MDS property, and with optimal disk access and network bandwidth.  This is shown in Figure~\ref{fig:dresscode}: an FR code is concatenated with an outer maximum distance separable (MDS) code.  
  Suppose that any $k$ nodes in the FR code contain at least $M$ distinct blocks between them.
%
We choose an outer MDS code which encodes $M$ data blocks $x_1,\ldots, x_M$ as $\theta$ encoded blocks $B_1,\ldots,B_\theta$, so that the concatenated code can recover the whole file $\vec{x}$ from any $k$ nodes.
%
For an FR code $\cal C$, let $M_{\cal C}(k)$ be the minimum number of unique data blocks contained in any set of $k \leq \alpha$ nodes.  
We call $M_{\cal C}(k)$ the \em storage capacity \em of $\C$.

\begin{figure}
\begin{center}
\footnotesize
\begin{tikzpicture}[scale=0.4]
\draw (0,0) rectangle (1,5);
\foreach \i in {1,...,4}
{
	\draw (0,\i) -- (1,\i);
}
\node at (0.5, 0.5) {$x_M$};
\node at (0.5, 4.5) {$x_1$};
\node at (0.5, 3.5) {$x_2$};
\draw[->] (1.5, 2.5) to node[above] {MDS code} (4.5, 2.5);
\begin{scope}[xshift=5cm,yshift=1cm]
	\draw (0,-2) rectangle (1,6);
	\foreach \i in {-1,...,6}
	{
	\draw (0,\i) -- (1,\i);
	}
	\node at (0.5, -1.5) {$B_\theta$};
	\node(b1) at (0.5, 5.5) {$B_1$};
	\node(b2) at (0.5, 4.5) {$B_2$};
\end{scope} 
\node at (8, 7) {FR code};
\node at (8, 3) {\begin{minipage}{1.5cm} \begin{center} \footnotesize Each block is replicated on $\rho$ nodes \end{center} \end{minipage} };
\begin{scope}[xshift=10cm,yshift=1cm]
\foreach \i in {-2, 0, 2, 4}
{
 \draw[thick] (0,\i) rectangle (2, \i + 1.5);
}
\draw (.5,4.2) rectangle (1.5,5.2);
\node(a1) at (1, 4.7) {$B_1$};
\draw[->] (b1.east) --(a1);
\draw (0,2.3) rectangle (1,3.3);
\node(c1) at (.5, 2.8) {$B_1$};
\draw[->] (b1.east) --(c1);

\draw [decorate,decoration={brace,amplitude=10pt,mirror},xshift=4pt,yshift=0pt]
(2,-2) -- (2,5.5) node [black,midway,xshift=1.5cm] {\begin{minipage}{2cm}\begin{center}$n$ nodes,\\ each storing $\alpha$ blocks \end{center}\end{minipage}};

\end{scope}
\end{tikzpicture}
\end{center}
\caption{A DRESS code is obtained by concatenating an outer MDS code with an FR code.  The FR code replicates each block $B_i$ $\rho$ times, and stores each copy on a separate storage node; each node holds $\alpha$ blocks.}
\label{fig:dresscode}
\end{figure}
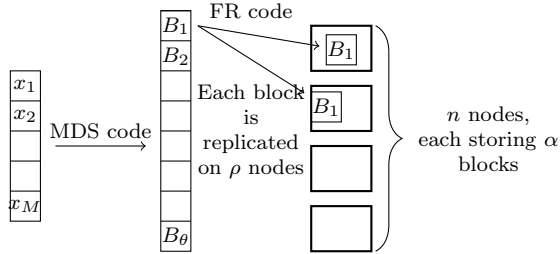

The goal is to store as much information as possible: that is, to make $M$ as large as possible given $k, \alpha,$ and $\rho$.  This parameter $M$ is called the \em storage capacity \em of the system.  
The fundamental trade-offs here are well-studied.  In the case where \em any \em $\alpha$ nodes should be able to repair a failed node, this is a special case of Minimum Bandwidth Regenerating codes, and the cut-set bound~\cite{dimakis2010network} gives an upper bound on $M$.  FR codes which meet or exceed this bound are called \em universally good \em\cite{el2010fractional}.

\subsection{Motivation}\label{sec:motivation}


\begin{figure}
\begin{center}
\footnotesize
\begin{tikzpicture}[scale=.33]
\begin{scope}
\begin{scope}
\draw (0,0) rectangle (1,2);
\node at (.5, .5) {2};
\node at (.5, 1.5) {1};
\node at (.5, 2.5) {$N_1$};
\draw[red,thick] (-.2,2.2) to (1.2, -.2);
\end{scope}
\begin{scope}[xshift=1.2cm]
\draw (0,0) rectangle (1,2);
\node at (.5, .5) {4};
\node at (.5, 1.5) {1};
\node at (.5, 2.5) {$N_2$};
\end{scope}
\begin{scope}[xshift=2.4cm]
\draw (0,0) rectangle (1,2);
\coordinate(y) at (.5,0);
\node at (.5, .5) {6};
\node at (.5, 1.5) {1};
\node at (.5, 2.5) {$N_3$};
\end{scope}
\begin{scope}[xshift=3.6cm]
\draw (0,0) rectangle (1,2);
\node at (.5, .5) {3};
\node at (.5, 1.5) {2};
\node at (.5, 2.5) {$N_4$};
\end{scope}
\begin{scope}[xshift=4.8cm]
\draw (0,0) rectangle (1,2);
\coordinate(x) at (.5,0);
\node at (.5, .5) {6};
\node at (.5, 1.5) {2};
\node at (.5, 2.5) {$N_5$};
\end{scope}
\begin{scope}[xshift=6.0cm]
\draw (0,0) rectangle (1,2);
\node at (.5, .5) {4};
\node at (.5, 1.5) {3};
\node at (.5, 2.5) {$N_6$};
\end{scope}
\begin{scope}[xshift=7.2cm]
\draw (0,0) rectangle (1,2);
\node at (.5, .5) {5};
\node at (.5, 1.5) {3};
\node at (.5, 2.5) {$N_7$};
\end{scope}
\begin{scope}[xshift=8.4cm]
\draw (0,0) rectangle (1,2);
\node at (.5, .5) {5};
\node at (.5, 1.5) {4};
\node at (.5, 2.5) {$N_8$};
\coordinate(w) at (.5, 0);
\end{scope}
\begin{scope}[xshift=9.6cm]
\draw (0,0) rectangle (1,2);
\node at (.5, .5) {6};
\node at (.5, 1.5) {4};
\node at (.5, 2.5) {$N_9$};
\draw[red,thick] (-.2,2.2) to (1.2, -.2);
\end{scope}

\begin{scope}[xshift=0cm,yshift=-5cm]
\draw (0,0) rectangle (1,2);
\node at (.5, .5) {2};
\node at (.5, 1.5) {1};
\coordinate (z1) at (.5, 2);
\coordinate (z2) at (1, 2);
\node at (.5, -.7) {$N_1'$};
\draw[->] (y) to node[pos=0.5, above] {1} (z1);
\draw[->] (x) to node[pos=0.5, above] {2} (z2);
\end{scope}

\begin{scope}[xshift=9.6cm,yshift=-5cm]
\draw (0,0) rectangle (1,2);
\node at (.5, .5) {6};
\node at (.5, 1.5) {4};
\coordinate (z1) at (.5, 2);
\coordinate (z2) at (0, 2);
\node at (.5, -.7) {$N_9'$};
\draw[->] (w) to node[pos=0.5, right] {4} (z1);
\draw[->] (x) to node[pos=0.5, above] {6} (z2);
\end{scope}

\node at (5.5,-6) {(a)};

\end{scope}
\begin{scope}[xshift=12cm]
\begin{scope}
\draw (0,0) rectangle (1,2);
\node at (.5, .5) {2};
\node at (.5, 1.5) {1};
\node at (.5, 2.5) {$N_1$};
\draw[red,thick] (-.2,2.2) to (1.2, -.2);
\end{scope}
\begin{scope}[xshift=1.2cm]
\draw (0,0) rectangle (1,2);
\node at (.5, .5) {4};
\node at (.5, 1.5) {1};
\node at (.5, 2.5) {$N_2$};
\end{scope}
\begin{scope}[xshift=2.4cm]
\draw (0,0) rectangle (1,2);
\node at (.5, .5) {6};
\node at (.5, 1.5) {1};
\node at (.5, 2.5) {$N_3$};
\coordinate(y) at (.5, 0);
\end{scope}
\begin{scope}[xshift=3.6cm]
\draw (0,0) rectangle (1,2);
\node at (.5, .5) {3};
\node at (.5, 1.5) {2};
\node at (.5, 2.5) {$N_4$};
\coordinate(x) at (.5, 0);
\end{scope}
\begin{scope}[xshift=4.8cm]
\draw (0,0) rectangle (1,2);
\node at (.5, .5) {5};
\node at (.5, 1.5) {2};
\node at (.5, 2.5) {$N_5$};
\draw[red,thick] (-.2,2.2) to (1.2, -.2);
\end{scope}
\begin{scope}[xshift=6.0cm]
\draw (0,0) rectangle (1,2);
\node at (.5, .5) {4};
\node at (.5, 1.5) {3};
\node at (.5, 2.5) {$N_6$};
\end{scope}
\begin{scope}[xshift=7.2cm]
\draw (0,0) rectangle (1,2);
\node at (.5, .5) {6};
\node at (.5, 1.5) {3};
\node at (.5, 2.5) {$N_7$};
\end{scope}
\begin{scope}[xshift=8.4cm]
\draw (0,0) rectangle (1,2);
\node at (.5, .5) {5};
\node at (.5, 1.5) {4};
\node at (.5, 2.5) {$N_8$};
\coordinate(w) at (.5, 0);
\end{scope}
\begin{scope}[xshift=9.6cm]
\draw (0,0) rectangle (1,2);
\node at (.5, .5) {6};
\node at (.5, 1.5) {5};
\node at (.5, 2.5) {$N_9$};
\end{scope}

\begin{scope}[xshift=0cm,yshift=-3cm]
\draw (0,0) rectangle (1,2);
\node at (.5, .5) {2};
\node at (.5, 1.5) {1};
\coordinate (z1) at (1, 1.8);
\coordinate (z2) at (1, 1);
\coordinate (xx) at (1,.25);
\node at (.5, -.7) {$N_1'$};
\draw[->] (y) to[out=-90,in=0] node[pos=0.5, above] {1} (z1);
\draw[->] (x) to[out=-90,in=0] node[pos=0.5, below] {2} (z2);
\end{scope}

\begin{scope}[xshift=4.8cm,yshift=-6cm]
\draw (0,0) rectangle (1,2);
\node at (.5, .5) {5};
\node at (.5, 1.5) {2};
\coordinate (z1) at (1, 2);
\coordinate (z2) at (0, 1.5);
\node at (.5, -.7) {$N_5'$};
\draw[->] (w) to node[pos=0.5, right] {5} (z1);
\draw[->] (xx) to node[pos=0.5, above] {2} (z2);
\end{scope}
\node at (2,-6) {(b)};
\end{scope}
\end{tikzpicture}
\end{center}
\caption{(a) A (9,2,3)-FR code repairs $N_1$ and $N_9$.  (b) A (9,2,3)-\abbr code repairs $N_1$ and $N_5$.  Both satisfy $|N_i \cap N_j| \leq 1$ for all $i \neq j$ and hence are universally good.}
\label{fig:frexamples2}
\end{figure}
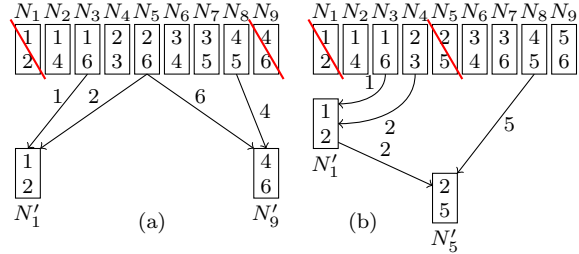

Before formally defining \abbr codes, we first motivate them by describing the property that we want them to have.
Consider the problem of greedily repairing multiple failed nodes in a system.
That is, if a node $N_i$ fails, we would like to repair it by greedily searching the system for any existing copies of the blocks that $N_i$ needs, without taking into account the full set $F$ of failed nodes. In this model, standard FR codes have some drawbacks, which we illustrate below. 

Figure~\ref{fig:frexamples2}(a) shows a universally good $(9,2,3)$-FR code with $9$ nodes and $6$ data blocks.
If a single node (say, $N_1$) fails, then it can be repaired by downloading one block from each of two different nodes.
However, suppose that both $N_1$ and $N_9$ fail simultaneously, and that $N_1$ is greedily repaired first using $N_3$ and $N_5$.  
Then there are no remaining copies of block $B_6$ in unused nodes, and either $N_3$ or $N_5$ (in the figure, it is $N_5$) must send a second block to $N_9'$ to repair $N_9$.

Is it possible to come up with a code which allows for the repair of $\rho-1$ failures, downloading just one symbol from each surviving block?
It is not hard to see that, in this model, the answer is no: suppose that the $\rho -1$ failed nodes were all holding block $B_1$.  Then all of these $\rho-1$ would need to contact the same surviving node to obtain $B_1$.
However, if the repair is allowed to be \em sequential\em---meaning that previously repaired nodes are allowed to be used as helper nodes themselves---then in fact this stronger load-balancing property is possible, even with greedy repair.

\abbr codes (defined below) are defined precisely to allow for greedy sequential repair of multiple failed nodes with disjoint repair groups.
Figure~\ref{fig:frexamples2}(b) shows an \abbr code with the same parameters as the FR code in Figure~\ref{fig:frexamples2}(a).   
If nodes $N_1$ and $N_9$ fail, it can be checked that no matter how $N_1$ is repaired, $N_9$ can still be repaired using a disjoint repair group (we will explain how to verify this in Section~\ref{sec:lbfr}), and this holds for any pair of nodes with $|N_i \cap N_j| = 0$.
The other case 
(shown in Figure~\ref{fig:frexamples2}(b)) is when two nodes with overlapping contents fail---say, $N_1$ and $N_5$, which share block $B_2$.  In this case, there is no non-sequential way to repair both nodes with disjoint repair groups, since only one copy of block $B_2$ remains.  However, if we are allowed to use $N_1'$ as a helper node after it is repaired, then we \em can \em repair $N_5$ with disjoint repair groups, as shown.

Our main result is that good \abbr codes exist and moreover
 that there are some parameter regimes where the parameters of FR codes match those of \abbr codes. 
%
There are (at least) two reasons why the LBFR property is desirable.

\begin{enumerate} 
\item \textbf{Reading less from disk when repairing multiple nodes.}  Our first motivation is that the load-balancing property, even if it comes at the cost of sequential repair, could result in reading less from disk.
For example, in Figure~\ref{fig:frexamples2}(a), $N_5$ must read two blocks, $B_2$ and $B_6$.
However in Figure~\ref{fig:frexamples2}(b), $N_5$ reads only one block, which is sent to $N_1'$.  Now, $N_1'$ does not need to first write and then read this block before sending it on to $N_5'$; rather it can send and then write.\footnote{Moreover, the situation depicted in Figure~\ref{fig:frexamples2}(a) can be adapted to the parallel setting, with the guarantee that $N_4$ reads only one block which it sends to both $N_1'$ and $N_5'$.} This saves a single disk access in this set-up, and in general when making $\rho - 1$ repairs it can save as many as $\rho - 2$ disk accesses.
\item \textbf{Faster repair with subpacketization.}
Our second motivation, which we discuss more in Section~\ref{sec:subpack}, comes from the speed of repair schemes in a particular communication model.   Suppose that each data block is actually made of of $T$ packets, and each node can send one packet per time step.  If a failed node (say, $N_1'$ in Figure~\ref{fig:frexamples2}(b)), is allowed to begin forwarding information before it is completely repaired, then it can forward the first packet in block $2$ to $N_5'$ before it has received the rest of the packets.  This allows for a ``pipelined'' repair which can be faster.  We show in Section~\ref{sec:subpack} that in this model, LBFR codes can complete repair faster than general FR codes in the worst case.
\end{enumerate}

%
Finally, we note that existing constructions of FR codes do not in general have this \abbr property.  For example, it is not hard to check that the FR codes obtained from Steiner systems in~\cite{el2010fractional} (for example, shown in Figure 5 in that work) are not \abbr codes. 

\subsection{Results}
The following are the main results of this work.
\begin{enumerate}
\item 
Our main technical contribution is Theorem~\ref{theorem:frpgraphs}, which characterizes \abbr codes precisely in terms of their \em adjacency graphs. \em    More precisely, we show that \abbr codes are those whose adjacency graphs contain no cycles of length four or six.  We note that Dimakis et al.~\cite{dimakis2014batch} use graphs with this property in a different way to obtain \em batch codes. \em  We discuss the relationship more in Section~\ref{sec:relatedwork}.  
\item  
In Theorem~\ref{theorem:good}, we show that \abbr codes are automatically \em univerally good, \em if they exist.  That is, they obtain storage capacity at least as good (in fact, strictly better) than the best Minimum Bandwidth Regenerating (MBR) codes.  

\item  In Section~\ref{sec:subpack}, we further explore the application of LBFR codes to speeding up repair in the presence of sub-packetization.  We show that in an appropriate model where each block is made up of $T$ packets and one packet can be sent in each timestep, it takes time $T(1 + o(1))$ timesteps (where $T \to \infty$) to complete a repair with LBFR codes.  
In contrast, there are FR codes where the number of timesteps required is greater by at least a constant factor.

\item Finally, in Section~\ref{sec:frpconstruct}, we give explicit constructions of \abbr codes---based on our characterization and existing constructions of girth-eight graphs---and establish that there are parameter regimes in which \abbr codes exist with the same storage capacity as optimal FR codes, making them optimal as well.
\end{enumerate}

\subsection{Outline}
After a brief discussion of related work in Section~\ref{sec:relatedwork}, we formally define \abbr codes and prove our main theorems about them in Section~\ref{sec:lbfr}.  Theorem~\ref{theorem:frpgraphs} characterizes \abbr codes in terms of their adjacency graphs.  Theorem~\ref{theorem:good} shows that \abbr codes are universally good.  In Section~\ref{sec:subpack} we discuss the details of repair algorithms for \abbr codes with sub-packetization.
Finally, in Section~\ref{sec:frpconstruct}, we give explicit constructions of \abbr codes.

\section{Background and related work}\label{sec:relatedwork}
Since FR codes were introduced by El Rouayheb and Ramchandran in~\cite{el2010fractional}, there have been several constructions in that paper and follow-up works extending the possible parameter regimes; see~\cite{olmez2016fractional} and the references therein.
The goal of these works is to create FR codes $\mathcal{C}$ in a variety of parameter regimes which lead to DRESS codes with the largest possible storage capacity $M_{\C}(k)$. 
To understand what the limitations on $M_{\C}(k)$ are, \cite{el2010fractional} considered the cut-set bound of~\cite{dimakis2010network}, a general lower bound which applies in particular to Minimum Bandwidth Regenerating (MBR) codes.
%
%
The cut-set bound implies that the storage capacity $C_{MBR}(n,k,\alpha)$ for MBR codes is equal to  $k\alpha - {k \choose 2}.$
The work \cite{el2010fractional} defined an FR code $\C$ to be \em universally good \em if $M_{\C}(k) \geq C_{MBR}(n,k,\alpha)$.
However, it is possible for $M_{\C}(k)$ to be strictly greater than $C_{MBR}(n,k,\alpha)$.  This is because FR codes allow for \em table-based \em repair: they guarantee that there exists a set of $\alpha$ helper nodes which can repair any failure, while MBR codes have the stronger guarantee that \em any \em set of $\alpha$ nodes will work. 
%
Given that the storage capacity for FR codes can exceed $C_{MBR}(n,k,\alpha)$, it is natural to ask what the optimal storage capacity is.  The work of~\cite{el2010fractional} provides two upper bounds.  The tighter of the two is a recursive upper bound, which establishes that
%
$M_{\C}(k) \leq g(k)$, where $g(k)$ is defined recursively by
$g(1) = d$ and 
\[g(k+1) = g(k) + d - \left \lceil \frac{\rho g(k) -kd}{n-k}\right \rceil.\]
The work of~\cite{silberstein2015optimal} gave constructions of FR codes matching this bound, showing that it is tight in some parameter regimes. 
Following~\cite{silberstein2015optimal}, we say that an FR code $\C$ is \emph{$k$-optimal} if $M_{\C}(k) = g(k)$, and \em optimal \em if this holds for all $k \leq \alpha$.


There are many constructions of FR codes in a wide range of parameter regimes.
In~\cite{zhu2014low}, universally good FR codes are constructed from Group Divisible Designs and in~\cite{park2016construction} construction methods based on perfect difference families and quasi-perfect  difference families are used to achieve more parameters for FR codes. 

More recent work, including~\cite{zhu2014replication,yu2014irregular,zhu2015heterogeneity}, has focused on constructing FR codes for systems with heterogeneous repetition degree, node storage capacity, and link costs. Yu et al.~\cite{yu2014irregular} address the problem of heterogeneity in the storage capacities and usage costs for nodes. They also consider network links with variable bandwidth, cost and rate. They define Irregular Fractional Repetition codes, and show the problem is an integer linear program. Zhu et al.~\cite{zhu2014replication} consider a model in which packets have one of two possible repetition degrees, motivated by a desire to have more popular packets more available due to higher repetition. They use a systematic outer MDS code, and give the data blocks the higher repetition degree while parity blocks have lower repetition, a strategy to reduce overall overhead in data reconstruction. More generally, in~\cite{zhu2015heterogeneity} repetition degrees can take any value greater than or equal to some value $\rho$, and node capacities also vary. 

Some work has been done on how FR codes can scale to larger systems. In~\cite{koo2011scalable}, constructions of cage graphs are used to create larger FR codes, in particular for when nodes have a much larger capacity than the repetition degree of the system. The code construction in~\cite{olmez2013constructions} uses Kronecker products to build larger FR codes. In~\cite{xu2013expander}, expander graphs are used to create and scale a code, though not an FR code. Most recently, in~\cite{itani2017dynamic} a genetic algorithm approach was used to solve for optimal block distributions on a system with possible data and node arrivals and heterogeneous block transfer costs.

There are also many notions in coding theory which are related to \abbr codes, including batch codes, PIR codes, Fractional Repetition Batch (FRB) codes, and uniform Erasure Combinatorial Batch Codes (ECBCs). 
The most similar to our work are FRB codes and Uniform ECBCs~\cite{silberstein2015fractional}.  Both of these notions
require that any $t$ data symbols can be recovered from $t$ distinct nodes.
ECBCs guarantee this even when $\rho-1$ nodes have failed. FRB codes have this batch property in addition to being an FR code and guaranteeing the file can be reconstructed from any $k$ nodes, but the batch property does not necessarily hold in the presence of failures.
Thus neither solve our motivating problem of doing repair in a load-balanced way when multiple nodes have failed.

Pliable Fractional Repetition Codes have also been recently studied, in which node storage capacity and block repetition degree can change over time~\cite{su2018pliable,su2018optimal}. In~\cite{su2018optimal}, the girth of the adjacency graph is used, defined as in our paper.  Those works apply the girth to analyze storage capacities, showing  conditions on girth and its relationship to storage capacity under which file size reaches optimal. They also show how a particular construction of Pliable Fractional Repetition Codes is also a Fractional Repetition Batch Code.

The specific case of graphs without four- or six-cycles have been studied before in the context of coding theory.  Dimakis et al.~\cite{dimakis2014batch} use graphs with this property to construct batch codes.  However, both their use of these graphs and the parameter regime that they are interested is quite from ours, reflecting the fact that their goal (primitive multiset batch codes) is different from our goal of LBFR codes.  In particular, they use these bipartite graphs to define parity checks for their codes, while we use them to determine how to replicate data.

Finally, there has been work on sequentially repairing multiple failures in distributed storage~\cite{prakash2014codes}.
Such work is in the context of locally repairable codes, and to the best of our knowledge, sequential repair has not been studied for FR codes before. 
%

%


\section{\abbr codes}\label{sec:lbfr}
We begin with some notation in order to formally define \abbr codes.
Suppose a set of nodes $F = \{N_1,N_2,...,N_{\rho-1}\}$ fails.  
%
A \emph{repair sequence} ${\cal S}$ for $F$ is defined as a sequence of blocks sent in the process of repairing the nodes in $F$. 
Formally, we write
 ${\cal S} = [(B_j,\R(B_j,N_i))]_{j = 1}^{\alpha(\rho-1)}$, where $\R(B_j,N_i)$ is the node which sends data block $B_j$ to the replacement node $N_i'$ for $N_i$ in the sequence.  We assume that the blocks corresponding to a node $N_i \in F$ are grouped together in $\mathcal{S}$, 
so $N_i = \{B_{(i-1)\alpha +j}: 1\leq j \leq \alpha\} $. 
Thus, each failed node corresponds to a contiguous group of $\alpha$ blocks in the sequence.
We define the \emph{repair group} for a node in a repair sequence $\cal S$ as $\Rg(N_i) = \{\R(B_j,N_i) \,|\, B_j \in N_i\}$. A repair sequence is \emph{valid} if for any $N_i$ that failed, $\R(B_j,N_i)\neq N_l \in F$ for $l>i$, meaning the sequence does not require use of a failed and un-repaired node.   We say that a repair sequence has \em disjoint repair groups \em if the sets $\Rg(N_i)$ are disjoint for $i \neq j$, meaning that no helper node has to send information to more than one replacement node.
%
With this notation in place, we define an \abbr code as follows.
\begin{definition}\label{def:frplus}
A $\FRp$ \abbr code $\C$, is a set of subsets $N_1,...,N_n$ of the set $\{B_1,...,B_{\theta}\}$ such that $n\alpha = \theta \rho$ where each $N_i$ for $i \in [n]$ corresponds to a node and each $B_j$ for $j \in [\theta]$ corresponds to a data block. $\C$ then has the properties:
 \begin{enumerate}
 \item $|N_i| = \alpha$ for $i = 1,...,n$; \label{frdef:size}
\item For each $B_j$, $j \in [\theta]$, $|\{N_i: B_j \in N_i\}| = \rho$;  \label{frdef:rep}
\item For any multi-set of $\rho-1$ failed nodes $\{N_1,...,N_{\rho-1}\}$,
and for all $1 \leq t < \alpha(\rho-1)$, the following holds:
For any valid repair sequence ${\cal S}_t = [B_j,\R(B_j,N_{\lceil j/\alpha\rceil})]_{j = 1}^t$ such that no node $\R(B_j,N_{\lceil j/\alpha\rceil})$ is repeated, there exists a node $\R(B_{t+1},N_{\lceil (t+1)/\alpha\rceil})$ not in ${\cal S}_t$ to repair the next block. 
\label{frdef:repair}
\end{enumerate}

\end{definition}

The first two properties are identical to those of Definition~\ref{def:fr}.  The third property guarantees precisely that \abbr codes can greedily sequentially repair any sequence of up to $\rho - 1$ failed nodes using disjoint repair groups. 

\subsection{Adjacency Graph Characterization}\label{ssec:graph}
We define the \em adjacency graph \em for an FR code (and hence also for an \abbr code) as follows. 
\begin{definition}
For a $(n,\alpha,\rho)$ fractional repetition code $\C$, the adjacency graph $G_{\C}$ is an undirected bipartite graph $({\cal V}_1,{\cal V}_2,E)$ such that $|{\cal V}_1|  = n$ and $|{\cal V}_2| = \theta$. Relate $G_{\C}$ to $\C$ by the function $v$ such that $v({N_i})\in {\cal V}_1$ corresponds to node $N_i$ and $v(B_j) \in {\cal V}_2$ corresponds to block $B_j$. Then $(v(N_i),v(B_j)) \in E(G_{\C})$ if and only if $B_j \in N_i$. 
\end{definition}
The adjacency graph $G_{\C}$ for FR code $\C$ is bi-regular with $d(v(N_i)) = \alpha$ for all $v(N_i) \in {\cal V}_1$ and $d(v(B_j)) = \rho$ for all $v(B_j) \in {\cal V}_2$.
Above, $d(v)$ denotes the degree of a vertex $v$ in the graph; $\Gamma(v)$ will denote the neighborhood of $v$ in $G_{\C}$.


\begin{thm}\label{theorem:frpgraphs}
An FR $\C$ described by adjacency graph $G_{\C}$ is a \abbr code if and only if $G_{\C}$ has no $4$- or $6$-cycles.\end{thm}
Theorem~\ref{theorem:frpgraphs} shows how we can check that the code shown in Figure~\ref{fig:frexamples2}(b) is indeed an \abbr code; indeed, it is straightforward to verify that the associated graph has no $4$-  or $6$-cycles.

We prove the two directions of Theorem~\ref{theorem:frpgraphs} in Lemmas~\ref{lemma:graphtoFR} and \ref{lemma:frptograph} below.
Before we prove the theorem, we note that this may be surprising.  Graphs without $4$- or $6$-cycles have been useful in constructing batch codes~\cite{dimakis2014batch}; however, those connections are not a tight characterization.  In contrast, Theorem~\ref{theorem:frpgraphs} shows that the load-balancing property that we are after---a very operational definition---is in fact \em equivalent \em to this very combinatorial statement. 
\begin{lemma}\label{lemma:graphtoFR}
Let $G_{\C} = ({\cal V}_1,{\cal V}_2,E)$ be a bi-regular bipartite graph without $4$- or $6$-cycles, such that $|{\cal V}_1| = n$, $|{\cal V}_2| =\theta$, with left degree $\alpha$ and right degree $\rho$. Then $G_{\C}$ is the adjacency graph of a $\FRp$ \abbr code $\C$.
\end{lemma}
\begin{proof}
Suppose that $G_{\C}$ is as above.
Properties~\ref{frdef:size} and~\ref{frdef:rep} of \abbr codes follow for $\C$ directly from the definition of $G_{\C}$. 

We next prove Property~\ref{frdef:repair} for some multi-set $F = \{N_1,...,N_{\rho-1}\}$ of failed nodes. Suppose that $B_1$ is the first block needed by the first node $N_1 \in F$. Since $|\Gamma(v(B_1))| = \rho$ and at most $\rho-1$ of these neighbors are unavailable to repair $N_1$, there exists a valid node for $\R(B_1,N_1)$.
Now we proceed by induction.
Suppose the repair sequence ${\cal S}_t = [B_j,\R(B_j,N_{\lceil j/\alpha \rceil})]_{j = 1}^t$ for some $t<\alpha(\rho-1)$ is valid and has 
disjoint repair groups.  
Let $i = \lceil(t+1)/\alpha\rceil$, so the sequence would be continued by transferring block $B_{t+1}$ to repair $N_i$. We show that at most $\rho-1$ locations of $B_{t+1}$ are unavailable to be $\R(B_{t+1},N_{i})$ (including $N_i$ itself), leaving at least one node available such that Property~\ref{frdef:repair} holds.

First, we show that if $B_{t+1}$ is not the first block of $N_i$  to appear in the sequence (i.e. $t+1> \alpha(i-1) +1$) no node $N$ such that $B_{t+1} \in N$ has been used to repair an earlier block for $N_i$.
Indeed, suppose that there were such a block $B_h$ that appeared before $B_{t+1}$ in $N_i$; then the node $N$ also contains $B_{t+1}$ and $B_h$.
But then $v(N)v(B_h)v(N_i)v(B_{t+1})$ is a $4$-cycle, which is a contradiction.

Now suppose that in fact, each of the $\rho-1$ locations of $B_{t+1}$ other than $N_i$ is unavailable, so by the previous step each must be in $(\bigcup_{\ell \in[i-1]} \Rg(N_\ell))\cup \{N_{i+1},...,N_{\rho-1}\}$. However, if two nodes corresponding to vertices in $\Gamma(v(B_{t+1}))$ were in $\Rg(N_\ell)$ for some $N_\ell \neq N_i$, the adjacency graph would have a $6$-cycle. Then at most $i-1+((\rho-1)-i) = \rho-2$ instances of $B_{t+1}$ are accounted for in nodes other than $N_i$. Thus there is some node in $\Gamma(v(B_j))$ which is neither already used in ${\cal S}_t$ nor not yet repaired, and this node can be used to transfer $B_j$ to $N_i$ in repair. 
 This establishes the inductive step and hence Property~\ref{frdef:repair}.
%
\end{proof}

\begin{lemma}\label{lemma:frptograph}
Let $\C$ be a $\FRp$-\abbr code. Then the adjacency graph $G_{\C}$ contains no $4$- or $6$-cycles.
\end{lemma}
\begin{proof}
First suppose a $4$-cycle exists in the adjacency graph $G_{\C}$ for $\C$. Let $B_1$ and $B_2$ the blocks on the $4$-cycle and $N_1$ and $N_2$ be the nodes. Now suppose all nodes containing $B_2$ except for $N_1$ fail, including $N_2$. Then $N_2$ also needs $B_1$ during repair. Greedily repair $N_2$ first, using $N_1$ for $B_1$. Then there is no way to recover $B_2$ and the graph is not an LBFR adjacency graph because \abbr Property~\ref{frdef:repair} does not hold.

Now suppose a $6$-cycle exists in $G_{\C}$, with blocks $B_1,B_2$, and $B_3$ and nodes $N_1,N_2$ and $N_3$. Suppose the $\rho-2$ neighbors of $B_3$ outside the cycle fail, as well as $N_1$. Repair $N_1$ using $N_2$ for $B_1$ and $N_3$ for $B_2$. Then $\rho-2$ failed nodes remain which need $B_3$, but no unused node has $B_3$ since we used $N_2$ and $N_3$ already. Thus the graph is not an LBFR adjacency graph because \abbr Property~\ref{frdef:repair} does not hold.
 \end{proof}
 \begin{figure}[h]
\centering
a.\includegraphics[width=0.21\textwidth]{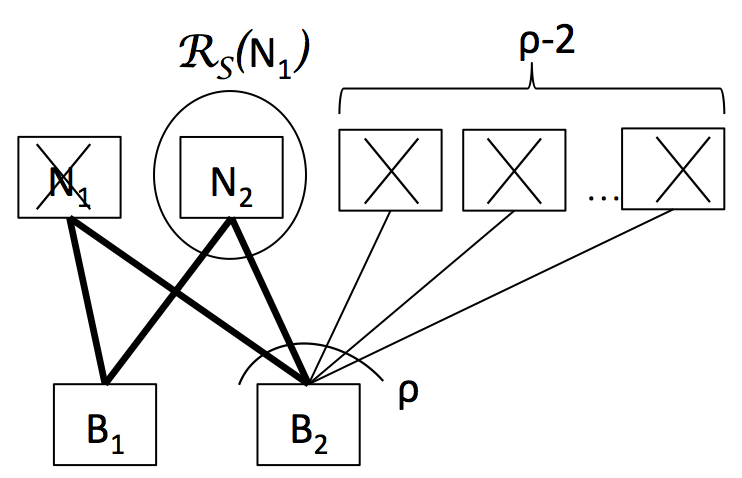}
b.\includegraphics[width=0.24\textwidth]{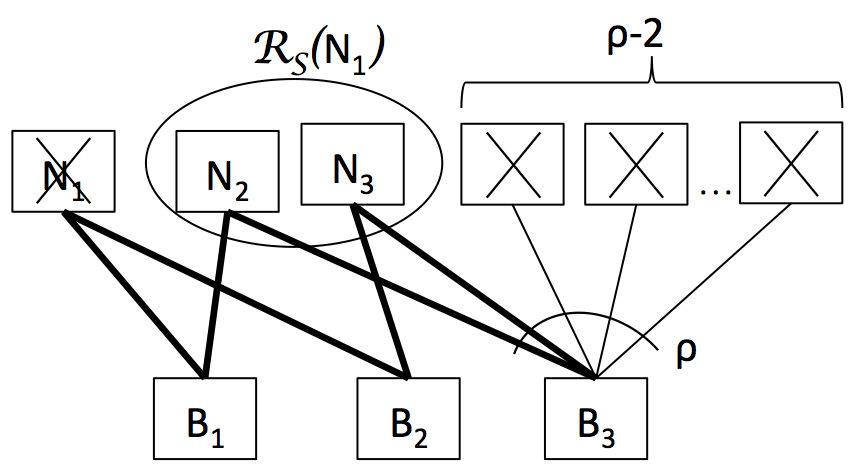}
\caption{Proof of Lemma~\ref{lemma:frptograph}: a. failure of \abbr Property~\ref{frdef:repair} with 4-cycle b. failure of \abbr Property~\ref{frdef:repair} with 6-cycle}
\label{fig:coderel}
\end{figure}

\subsection{Storage Capacity}\label{ssec:good}
In this section, we show that any \abbr code is universally good: that is, the storage capacity exceeds (in fact, strictly exceeds) that of any MBR code with the same parameters.
We recall that
\begin{equation}\label{eq:storage}
 M_{\C}(k) := \left|\min_{I \subset [n], |I| = k}{\bigcup_{i \in I} N_i}\right|
\end{equation}
denotes the storage capacity of an \abbr code $\C$.
%

\begin{thm}\label{theorem:good}
$M_{\C}(k) > C_{MBR}(n,k,d)= k\alpha - {k\choose 2}$.
\end{thm}
\begin{proof}
Let $\C$ be a \abbr code with adjacency graph $G_{\C}$. 
Suppose that $I$ minimizes the right hand side of \eqref{eq:storage}.  Without loss of generality and for notational convenience, we will assume $I = \{1,\ldots,k\}$.
Let $r$ denote the number of block duplications in $I$: that is, $r = \sum_i|N_i| - |\cup_{i \in I} N_i| = k\alpha -M_{\C}(k)$. 
We will show $r < {k \choose 2}$ for any $G$ without $4$ or $6$ cycles; this will imply that ${k \choose 2} > r = k\alpha - M_{\C}(k)$, and thus $M_{\C}(k) > C_{MBR}(n,k,\alpha)$.
As a result of Property~\ref{frdef:repair} in the definition of \abbr codes and Lemma~\ref{lemma:graphtoFR}, any two nodes in $\{N_1,\ldots,N_k\}$ share at most one block.  Thus the maximum number of duplications possible is $r \leq {k \choose 2}$.  
In fact, we must have $r < {k \choose 2}$.  To see this, assume that $r = {k \choose 2}$.  Then in fact any two nodes in $N_1,\ldots,N_k$ share exactly one block.
Consider any three nodes, w.l.o.g. suppose they are $N_1,N_2$, and $N_3$. Then if $\{B_1\} = N_1\cap N_2$, $\{B_2\} = N_2\cap N_3$, and $\{B_3\} = N_3\cap N_1$, for data blocks $B_1,B_2,B_3$, the cycle $v(N_1)v(B_1)v(N_2)v(B_2)v(N_3)v(B_3)v(N_1)$ is a 6-cycle.  This contradicts Theorem~\ref{theorem:frpgraphs}.
Thus $r < {k \choose 2}$, and so $k\alpha - M_{\C}(k) < {k \choose 2}$, as desired.
\end{proof}
\


\section{Supacketization and Repair Time}\label{sec:subpack}
So far, we have shown that LBFR codes are useful when we have discrete data blocks and do not want to ever have a node send more than one block. In this section we consider a new communication model, in which subpacketization is allowed. It turns out that the LBFR property is also useful here, and it allows for better repair time. 

We consider a model where a repaired node can begin acting as a helper node while it is being repaired.  If each block consists of $T$ packets, and each node can send only one packet at each timestep, then when $T$ is large the repair time for LBFR codes is $T(1 + o(1))$, assuming $n, \alpha, \rho$ are fixed and $T \to \infty$ (Theorem~\ref{thm:lbfrupper}). 
On the other hand, the repair time for a general FR code could be a constant factor larger than that (Theorem~\ref{thm:frnoforward}). 

\subsection{Communication Model}

Here, we specify the communication model under which the results described above hold.  We assume that each data block $B_1, \ldots, B_\theta$ consists of $T$ packets.  Time is discrete, and each packet can be transmitted in one unit of time.   Some nodes are designated as \em helper nodes, \em as will be described below.  In any time step, each helper node  may choose to send at most one packet that it holds to any other node in the system.  

As before, we consider a situation where there are $n$ nodes, so that node $i$ holds a set of data blocks with indices $N_i \subset [\theta]$ of data blocks.  We suppose that some set $S \subset [n]$ of nodes fail completely, losing all of their data.  Replacement nodes (also called $S$) are set up to recover this data.  The \em repair time \em of a repair scheme in is the number of time steps until the last node $i \in S$ recovers all of the data blocks indexed by $N_i$, where communication is according to the model described above.

We consider two cases of designating nodes as helper nodes.  In the first case, any nodes (including the replacement nodes $S$) can be helper nodes; the only requirement is that a node must receive data before sending it along.  In this case, we say that \em forwarding \em is allowed.  In the second case, only the nodes which are \em not \em replacement nodes (that is, nodes indexed by $[n] \setminus S$) can be helper nodes.  In this case, we say that forwarding is not allowed.

\subsection{Node Graph}
We next define a graph derived from the adjacency graph of  FR codes and LBFR codes, which will be used in analyzing repair time.
\begin{definition}
Let $\C$ be an LBFR with adjacency graph $G = (\V_1,\V_2,E)$. Define the \emph{node graph}, $\Gp$, of $\C$ as follows. Let $V(\Gp) = \V_1$. Let $E(\Gp) = \{(v_i,v_j)\in \V_1\times \V_1: N_i\cap N_j \neq \emptyset\}$. That is, there is an edge between nodes $i$ and $j$ if those two nodes have a block in common.
\end{definition}
\begin{definition}
 Let $\poS$ denote the outer boundary of a node set $S \subset V(G)$ for some graph $G$:
 $$\poS = \{v \in  V(G)\setminus S: \exists u \in S\text{ s.t.  } (u,v) \in E(G)  \}$$

\end{definition}
 Note that for graph $\Gp$ with some set of failed nodes $S$, $\poS$ represents the set of nodes in $V(\Gp)$ which can be used as helper nodes in the model where forwarding is not allowed.
We then define the following measure (similar to the Cheeger constant) over a graph $\Gp$:
\begin{definition}
Let $h_{\rho}(\Gp)$ be  defined as:
$$h_{\rho}(\Gp) = \displaystyle\min_{\substack{S\subseteq V(\Gp)\\1\leq|S|\leq  \rho-1}} \frac{|\partial_{out} S|}{|S|}$$
 \end{definition}
 We require that $|S|\leq \rho-1$ to guarantee that if the nodes of $S$ are in the set of a failed nodes, some repair sequence exists.  This guarantee comes from the fact that all blocks are replicated $\rho$ times, so up to $\rho-1$ node failures still leave a copy of each block available for repair.

\subsection{Lower Bounds for General FR Codes}\label{sec:frlower}
In order to motivate our upper bounds for LBFR codes, we first show lower bounds for general FR codes. We show that general FR codes $\mathcal{C}$ with small $h_{\rho}(G_{\mathcal{C}}')$ have a large worst-case repair time.  Our main theorem shows that this is the case when forwarding is not allowed.  However, as we show below, there exist FR codes where forwarding cannot help, which implies that FR codes can have a large worst-case repair time even if forwarding is allowed.

\begin{thm}\label{thm:frnoforward}
Let $\mathcal{C}$ be an $(n,\alpha,\rho)$-FR code with node graph $G_{\mathcal{C}}'$, and assume the communication model above with no forwarding, where each data block consists of $T$ packets.
Then for any repair algorithm, there exists a set $S \subseteq [n]$ of size at most $\rho -1$ so that repairing the failed set $S$ requires repair time at least
\[ T \cdot \left\lceil \frac{\alpha}{h_{\rho}(\Gp)} \right\rceil. \]
\end{thm}

In some FR codes, the restriction of no forwarding allowed can be removed because there is no opportunity to forward blocks due to failed nodes not sharing any required blocks. Thus this setting represents the worst case failure patterns and repair schemes for those codes under a model where forwarding \emph{is} allowed.
\begin{prop}\label{prop:nofwdpossible}
There exist $(n,\alpha,\rho)$-FR codes, for some $n, \alpha, \rho$,  in which some sets of $\rho-1$ node failures make forwarding impossible, because no failed nodes are requesting any of the same blocks.
\end{prop}
We prove Proposition~\ref{prop:nofwdpossible} by example in Section~\ref{sec:examples}, and it immediately implies the following corollary.

\begin{cor}\label{cor:frbad}
There exist $(n,\alpha,\rho)$-FR codes, for some $n, \alpha, \rho$, so that, even when forwarding is allowed, the worst-case bound of Theorem~\ref{thm:frnoforward} still holds.
\end{cor} 


To prove Theorem~\ref{thm:frnoforward}, we first show the following property of FR codes and repair sequences.
\begin{lemma}\label{lem:minsend}
Let $\Gp = (V',E')$ be the node graph for an $(n,\alpha,\rho$)-FR code,  $\C$. For $S\subset V'$ with $|S|\leq  \rho -1$,  if $\frac{\alpha|S|}{|\poS|} > x$ for some  $x \in \Z^+$, there exists a set $S'$ of failed nodes with $|S'|\leq \rho-1$ and $S\subseteq S'$ and a repair sequence $\cS_t = [B_j,R_\cS(B_j,N_{\lceil j/\alpha\rceil})]_{j=1}^t$ for some $1\leq t <\alpha|S'|$  such that any extension of the sequence (up to length $t= \alpha|S'|$) results in a node in $\poS$ transmitting a total of $x+1$ blocks.
\end{lemma}
\begin{proof}

Assume the nodes of $S'$ have failed, including the nodes of $S$. Assume we are repairing $S \subseteq S'$ first. Consider $\poS\setminus S'$, the set of nodes that have blocks that nodes in $S$ want. 
There are $|S|\alpha$ block transfers that must  be done to complete repair of $S$, so if $\frac{\alpha|S|}{|\poS\setminus S'|}\geq \frac{\alpha|S|}{|\poS|} > x$ by the pigeonhole principle at least one of the helper nodes in $\poS$ must  transmit more than $x$ blocks to complete a repair sequence $\cS_{\alpha(|S|)}$ (and thus also a full repair sequence which repairs $S'\setminus S$ after repairing $S$).


\end{proof}

Now we can prove Theorem~\ref{thm:frnoforward}.

\begin{proof}{(Theorem~\ref{thm:frnoforward})}
Let $S_{h}$ be the set which minimizes $\displaystyle\min_{\substack{S\subseteq V(\Gp)\\1\leq|S|\leq \rho-1}} \frac{|\poS|}{|S|}$, i.e. corresponding to $h_{\rho}(\Gp)$.  Since $|S_h|\leq  \rho-1$, we can apply Lemma~\ref{lem:minsend}: if $\frac{\alpha|S_h|}{|\partial_{out}(S_h)|} > x$ for some $x \in \mathbb{Z}^+$, there is a greedy repair outcome for a set of failed nodes that includes $S_h$ such that at least one node must transfer at least $x+1$ data blocks,  requiring time $(x+1)T$. Substituting the definition of edge expansion, $\frac{\alpha}{h_{\rho}(\Gp)} \geq x$ implies there exists a case of $(x+1)T$ transfer time, or equivalently, the worst case transfer time required is at least $T\cdot \lceil \frac{\alpha}{h_{\rho}(\Gp)}\rceil$.
\end{proof}

\subsection{Bounds on Repair Time for LBFR Codes} \label{sec:lbfrbounds}
Allowing failed nodes to forward blocks they have received during repair is only useful if it can reduce the total number of blocks transmitted by the original helper nodes. Here we show that the LBFR property guarantees that forwarding will be useful,  so we can get stronger bounds on  repair time than for general FR codes without forwarding, or cases of FR codes in which forwarding is allowed but not possible. We start with an upper bound on LBFR codes, followed by a lower bound which highlights the difference between LBFR codes and FR codes in general.

We first describe how  to derive an algorithm for an LBFR with forwarding, based on a greedy LBFR repair algorithm with $T=1$.
Recall that for any multi-set of $\rho-1$ failed nodes, there are $\alpha(\rho-1)$ block transfers which must be completed, so an algorithm must assign each of these transfers to a source node. A repair sequence is created by determining a helper for each  transfer, indexed by $t$ for $1\leq t \leq \alpha(\rho-1)$.  Property~\ref{frdef:repair} of an LBFR code (Definition~\ref{def:frplus}) states that at any such step $t$, with a  valid repair sequence ${\cal S}_t = [B_j,\R(B_j,N_{\lceil j/\alpha\rceil})]_{j = 1}^t$ such that no node $\R(B_j,N_{\lceil j/\alpha\rceil})$ is repeated, there exists a node $\R(B_{t+1},N_{\lceil (t+1)/\alpha\rceil})$ not in ${\cal S}_t$ to repair the next block. Consider the class of algorithms which at each step $t$ choose any node not already used and which has the requested block. These algorithms are guaranteed to all create a complete, valid repair sequence. 

Now suppose we have $T>1$ and forwarding is allowed. We first create a repair sequence as before, where all $T$ packets of the block $B_j$ are transferred at each step $j$ and each node is used as a helper for exactly one block (i.e. all $T$ of its packets). The repair sequence gives us an assignment of exactly one block of $T$ packets and receiver for each node which acts as a helper. To do the actual repair, each helper node begins doing its assigned transfer of packets as soon as possible. Nodes that have not failed can begin at the start of repair, and nodes that are also receiving the block will forward each of the $T$ packets as they are received.

 \begin{thm}\label{thm:lbfrupper}
Let $\C$ be an $(n,\alpha,\rho)$-LBFR. Any algorithm derived as above from a greedy LBFR repair algorithm has repair time at most $T + \rho - 2$.
\end{thm}
\begin{proof}

 Let $S$ be the set of failed nodes, with $|S| \leq \rho-1$. Consider the set of block transfers during repair for some block $i$. Since the LBFR property ensures that no node sends more than one block, no node can send $i$ to more than one receiver. Thus $i$ can only be transferred along paths through the node graph $\Gp$ that do not branch. Additionally, any such path must start at a node not in $S$, and then have all other path nodes in $S$. Since $|S|\leq \rho-1$, any path has at most $\rho-1$ edges. It then takes time at most $T+(\rho-2)$ to get a packet through such a forwarding path. Since no node is sending more than one block total, all paths can operate independently and the total transfer time will be at most $T+ (\rho-2)$.
 
 \end{proof}

We next have an impossibility result when no node sends packets from more than one block; notice that this is the case for the algorithm in the upper bound above.

\begin{thm}\label{thm:lbfrlower}
Let $\C$ be an $(n, \alpha, \rho)$-LBFR code, and assume the communication model described above, where forwarding is allowed.  Then there exists a set $S$ of at most $\rho - 1$ failures such that any repair algorithm in which no node sends packets from more than one block has repair time at least
\[ T+\left \lceil \frac{\alpha}{h_{\rho}(\Gp)}\right \rceil -1. \]
\end{thm}

To prove this, we first define the following function over blocks:
\begin{definition}\label{def:srcfunc}
Let $\mathcal{C}$ be an $(n,\alpha,\rho)$-LBFR  with $\theta$ blocks, a set $S$ of node failures, and some repair scheme. Let $\mathcal{T}\subset [\theta]\times[n]\times [n]$ be the set of block transfers in the repair scheme, where $(i,a,b)\in \mathcal{T}$  represents block $i$ sent from node $a$ to node $b$. Then we recursively define the \emph{source function} $s:[\theta]\times[n]\times[n]\to[n]$  as follows:
\[s((i,a,b)) =\begin{cases}
a & a \notin S\\
s((i,c,a)) \text{ s.t. } (i,c,a)\in\mathcal{T} & a \in S

\end{cases} \]
\end{definition}
The source function evaluated on the transfer of some block $i$  from node $a$ to node $b$, $s((i,a,b))$, is the unique node $c \notin S$ where the block $i$ originated. This is well-defined because there is exactly one $(i,c,a)\in \mathcal{T}$ for each $(i,a,b)\in \mathcal{T}$ if $a \in S$. If $a \notin S$, $a$ is not receiving the block from anyone since $a$ is only acting as a helper, so there is no corresponding $(i,c,a)$ and the recurrence terminates. Figure~\ref{fig:srcfunction} shows an example of this.
 \begin{figure}[h]
\centering
 \subfloat[][]{\includegraphics[width=0.45\textwidth]{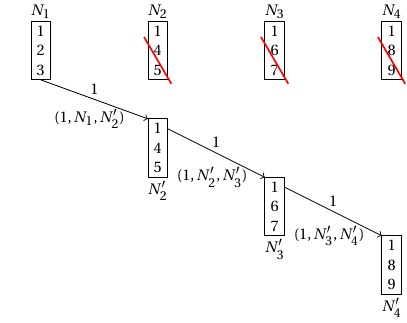}}
\\ 
 \subfloat[][]{
    \begin{tabular}{l|l}
\hline
     $(i,a,b) \in \mathcal{T}$ & $s((i,a,b))$ \\
    \hline
    $(1,N_1,N_2')$ & $N_1$  \\
     $(1,N_2',N_3')$ & $s((1,N_1,N_2')) = N_1$  \\
     $(1,N_3',N_4')$ & $s((1,N_2',N_3')) =s((1,N_1,N_2')) = N_1$  \\

    \hline
  \end{tabular}}
\caption{Example of computing the source function recursively. Only a subset of the FR code with $\rho = 4$ is shown, including failed nodes $S = \{N_2,N_3,N_4\}$ and their repair of block $1$, using helper $N_1$. For each block transfer in $\mathcal{T}$ in the chain, the function recurses until the value $N_1$ since $N_1 \in S$.}\label{fig:srcfunction}
\end{figure}


Now we can use the groups of repairs with the same source to lower bound the length of forwarding chains, and thus repair time.
\begin{proof}(Theorem~\ref{thm:lbfrlower})
Suppose we have any repair scheme  to repair a set $S$ of up to $\rho-1$ failed nodes. The corresponding set of block transfers $\mathcal{T} \subset   [\theta]\times[n]\times [n]$ then has $|\mathcal{T}| =\alpha |S| \leq \alpha(\rho-1)$. By definition of $s$, the range of the source function  $s$ evaluated on $\mathcal{T}$ is contained in $\poS$: 
\[\{a \in [n]: s(t) = a \text{ for }t \in \mathcal{T}\}\subseteq \poS.\]
  It follows from the pigeonhole principle that some set of at least $\frac{\alpha|S|}{|\poS|}$ transfers in $\mathcal{T}$ all have the same value under $s$: 
  \[\exists a \in \poS \text{ s.t. }|\{t \in \mathcal{T}: s(t) = a\}|\geq \left\lceil \frac{\alpha|S|}{|\poS|}\right\rceil\]

 Now suppose the set of failed nodes is $S_h$, which minimizes $\displaystyle\min_{\substack{S\subseteq V(\Gp)\\1\leq|S|\leq \rho-1}} \frac{|\poS|}{|S|}$, i.e. corresponding to $h_{\rho}(\Gp)$. Then we have  some  $a \in \partial_{out}(S_h)$ as above with  $|\{t \in \mathcal{T}: s(t) = a\}|\geq \lceil \frac{\alpha|S_h|}{|\partial_{out}(S_{h})|} \rceil= \lceil \frac{\alpha}{h_{\rho}(\Gp)}\rceil$. With the requirement that each node is only allowed to send at most one block, the only possible configuration of these $\lceil \frac{\alpha}{h_{\rho}(\Gp)}\rceil$ transfers is a single chain of forwarding without any branching. The chain begins with node $a$ sending the block to one of the failed nodes. The transfer of $T$ packets forwarded along this chain requires time $T + \lceil \frac{\alpha}{h_{\rho}(\Gp)}\rceil-1$.

\end{proof}

\subsection{Consequences of Theorems~\ref{thm:frnoforward},~\ref{thm:lbfrupper}, and~\ref{thm:lbfrlower}}

We compare LBFR codes to FR codes when applying the subpacketization technique described in Theorem~\ref{thm:lbfrlower}.
FR codes in general have lower bound on transfer time that remains fixed: 
\[ \text{repair time for any FR code (without forwarding)} \geq T\left\lceil \frac{\alpha}{h_{\rho}(\Gp)}\right\rceil.\]

On the other hand, in the forwarding model, LBFR codes can take advantage of the ability to forward: we have
\[T+ \left\lceil \frac{\alpha}{h_{\rho}(\Gp)}\right\rceil \leq \text{repair time for any LBFR code (with forwarding)} \leq T+(\rho-2). \]
In particular, when $T$ is large, the repair time is dominated by $T$ in the case of LBFR codes with forwarding,  regardless of code parameters $n,\alpha$, and $\rho$.  As we will show in Section~\ref{sec:examples}, there are cases where the upper and lower bound for LBFR codes are tight. The proofs of Theorems~\ref{thm:frnoforward} and~\ref{thm:lbfrupper} also illustrate how the distinction between LBFR codes and FR codes in general applies to models with forwarding and subpacketization. While the value of $\lceil \frac{\alpha}{h_{\rho}(\Gp)}\rceil$ lower bounds how many distinct blocks some node in the corresponding FR code must send, in an LBFR code it lower bounds the length of a chain of nodes forwarding a particular block.


\subsection{Examples}\label{sec:examples}
In this section, we give two examples.  The first is of a general FR code $\C$
for which forwarding is not possible, so the lower bound on repair time in Theorem~\ref{thm:frnoforward} (Section~\ref{sec:frlower}) applies even in the forwarding model; this proves Proposition~\ref{prop:nofwdpossible}. 
The second example is of an LBFR code which shows that our upper and lower bounds are tight. 

Our first example of a general FR code
is shown in Figure~\ref{fig:frnoforward}, $\C$ has $n=\theta = 7$, and $\alpha = \rho = 3$. If any two nodes which have no blocks in common fail, such as $N_1$ and $N_7$ in the picture, there are only $5$ nodes available as helpers. Since forwarding is not possible, clearly at least one of the five helpers must be used twice in the transfer of a total of $6$ blocks to complete repair. This results in a total repair time of at least $2T$. An example like this could also appear in a larger code, if other failures are repaired first leading to this pattern of available helpers at the end of the execution of some repair algorithm.

 \begin{figure}[h]
\centering
\includegraphics[width=0.3\textwidth]{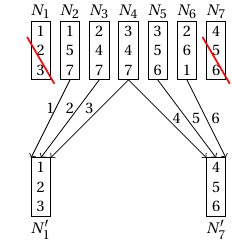}
\caption{An FR code with two failed nodes, such that repair of the failures cannot use forwarding. As a result, some helper must send two blocks for repair to be completed, and total repair time is at least $2T$.}\label{fig:frnoforward}
\end{figure}


Our second example is an LBFR code which achieves the lower bound in Theorem~\ref{thm:lbfrlower} (Section~\ref{sec:lbfrbounds}) for repair time. For this code the upper bound, Theorem~\ref{thm:lbfrupper}, is also tight. Let $\C$ be an FR code with the Tutte-Coxeter graph (the unique (3,8)-cage) as the adjacency graph. The Tutte-Coxeter  graph has 30 vertices and girth  $8$, so it specifies a $(15,3,3)$-LBFR, and we can analyze repair time allowing forwarding between failed nodes. Our lower bound evaluates to $T+\lceil\frac{2}{3}\rceil =  T+1$ and  the upper bound gives $T+ (3-2)  =  T +1$ so the bounds are tight for repair time in this case.


\section{Constructions}\label{sec:frpconstruct}

Theorem~\ref{theorem:frpgraphs} implies that any graph without $4$- or $6$-cycles gives rise to a \abbr code.  Such graphs are well-studied, and in this section we review two of these constructions and work out the implications for distributed storage.

The first construction, generalized quadrangles, have already been used to construct FR codes~\cite{silberstein2015optimal}.  This graphs already have girth $8$, which implies that the FR codes of \cite{silberstein2015optimal} are in fact LBFR codes.

The second construction, which we discuss in more detail, are the graphs of Lazebnik, Ustimenko and Woldar (LUW), \cite{lazebnik1994new}, which were also used in \cite{dimakis2014batch} in a different context to construct batch codes.  To the best of our knowledge, these graphs have not been used before to construct FR codes.  However, we show that LUW graphs yield LBFR codes which match the parameters of state-of-the-art FR code constructions arising from transversal designs~\cite{silberstein2015optimal}.

The take-away from these two examples is that in a variety of parameter regimes, there are constructions of LBFR codes which match the parameters of state-of-the-art FR codes.  Moreover, in many settings where optimal FR codes are known (in the sense that they meet the recursive bound of \cite{el2010fractional}), we also have optimal LBFR codes.  (We refer the reader to~\cite{olmez2016fractional} and the references therein for a more complete discussion of the parameters of existing FR codes.) 

%

We note that it is possible to extend the range of parameters for \abbr codes by taking covering graphs, or lifts, of graphs for smaller \abbr codes~\cite{gross1987topological}. Since lifts do not reduce girth, a lift of a  graph also describes an \abbr code. Constructing a lift also preserves repetition degree, node capacity, and disjoint repair sets, while linearly increasing numbers of data blocks and nodes. As a result, storage capacity (as a function of $k$ depending only on $\rho$ and $\alpha$) is also maintained.

\subsection{Generalized quadrangles}
\em Generalized quadrangles \em have been used in \cite{silberstein2015optimal} to construct FR codes with optimal storage capacity.   
A generalized quadrangle the case of a generalized $m$-gon for  $m = 4$. By definition,  generalized $m$-gon incidence graphs have girth of $2m$ so the incidence graph of a generalized quadrangle has girth of eight~\cite{godsil2013algebraic}.  Therefore, these FR codes are also automatically LBFR codes.

Generalized quadrangles can give optimal $((s+1)(st+1),(t+1),(s+1))$-LBFR codes in settings where $t \geq s$ so that an $(s,t)$ generalized quadrangle exists.  
It was shown in~\cite{silberstein2015optimal} (Theorem 31) that FR codes from generalized quadrangles are $k$-optimal whenever $k \leq 3$. As a result, they are optimal when $\alpha = 3$; this requires $s \leq t = 2$, and so there are only a small number of optimal examples.  However, this example does show that \abbr codes can be optimal.

%

\subsection{Lazebnik, Ustimenko, Woldar Graphs and Transversal Designs}
In this section we will describe how a construction of girth-eight graphs used for batch codes in \cite{dimakis2014batch} achieves the same storage capacity as the construction of FR codes using transversal designs in \cite{silberstein2015optimal}.  Transversal designs have adjacency graphs of girth exactly six, so they do not give LBFR codes~\cite{anderson1997combinatorial,silberstein2015optimal}. However, conditions given in \cite{silberstein2015optimal} for optimal FR codes from transversal designs can be applied to LBFR codes constructed using Lazebnik-Ustimenko-Woldar (LUW) graphs~\cite{lazebnik1994new}. 

As we elaborate on more below, the girth-eight LUW graphs give rise to $(q^{3+t},q^s,q^t)$-\abbr codes for some odd prime power $q$ and some $t\in (0,2]$, $s \in [0,1)$. In a minimal example of $q = 3$, $s = t = 1$, we have $81$ storage nodes, $81$ data blocks, and $\alpha = \rho = 3$. Without symmetric parameters, a small example is $q = 3^2$, $t = 1$, $s = 0.5$, for which we have $9^4 = 6561$ storage nodes, $9^{3.5} = 2187$ data blocks, $\rho = 9$, and $\alpha = 3$. These graphs have the advantage, compared to other constructions of LBFR codes from girth-$8$ graphs, of having parameters $(\theta,\alpha)$ independent from $(n,\rho)$ as we scale with $s$ or $t$. 


We first describe the construction of LUW graphs, and then show that  the storage capacity $M(k)$ of these codes is $M(k)\geq k\alpha-{k \choose 2} + br + \rho{b \choose 2}$, which matches the bound that \cite{silberstein2015optimal} obtain for FR codes using transversal designs.

\begin{definition}[\cite{lazebnik1994new}]\label{def:luw}
Let $s \in [0,1)$, $t \in (0,2]$ such that $q^t$ and $q^s$ are integers. Let ${\mathcal{T}}\subseteq \bF_{q^2}$ and $\mathcal{S}\subseteq \bF_q$ such that $|\mathcal{T}| = q^t$ and $|\mathcal{S}| = q^s$.  Let $L  = \{[\ell] =  (\ell_1,\ell_2,\ell_3)| \ell_1 \in \mathcal{T}\subseteq \bF_{q^2},\ell_2 \in \bF_{q^2},\ell_3 \in \bF_{q}\}$ and $P  = \{(p)=(p_1,p_2,p_3)| p_1 \in \mathcal{S}\subseteq \bF_q,p_2\in \bF_{q^2},p_3 \in \bF_q\}$. 
We define a bipartite graph $G = (\mathcal{V}_1,\mathcal{V}_2,E)$ as follows with $|\mathcal{V}_1| = q^{3+t}$ and $|\mathcal{V}_2|=q^{3+s}$:
\begin{enumerate}
 \item 
 Each vertex $[\ell] \in \mathcal{V}_1$ corresponds to to $[\ell] = (\ell_1,\ell_2,\ell_3) \in L$ and vertices $(p) \in \mathcal{V}_2$ correspond to $(p)= (p_1,p_2,p_3)\in P$. \label{itme:luwvertices}
 \item An edge $([\ell],(p)) \in E$ exists when both
$p_2 -\ell_2 = p_1\ell_1$ and $p_3 - \ell_3 =\overline{{p_1}}\ell_2 + p_1\overline{{\ell_2}}$. \label{item:luwedges}
\end{enumerate}
The function $x \mapsto \overline{x}$ denotes the involutive automorphism\footnote{The involutive automorphism $x \mapsto \bar{x}$ with fixed field $\mathbb{F}_q$ is the automorphism which fixes $\mathbb{F}_q$ and so that $\bar{\bar{x}} = x$.}  
$F_{q^2} \to F_{q^2}$ with fixed field $F_q$.
\end{definition}
Let $\{B_\ell^0, B_\ell^1,...,B_\ell^{q^t-1}\}$ denote the partition of $L$ such that for each $i = 0...,q^t-1$, all $[\ell] \in B_\ell^i$ have the same first coordinate.
\begin{thm}
For any $q,s,t,\mathcal{S}$, and $\mathcal{T}$ satisfying the requirements of Definition~\ref{def:luw}, let $G$ be a graph constructed according to the definition. Then any \abbr code with $G$ as the adjacency graph will have storage capacity $M(k) \geq k\alpha - {k \choose 2} + \rho {b \choose 2} + br $, where $k = b\rho  + r $ for $r < \rho$.
\end{thm}
\begin{proof}
Let $\mathcal{C}$ be a $(q^{3+t},q^s,q^t)$-LBFR code with adjacency graph $G = (\mathcal{V}_1,\mathcal{V}_2,E)$, where $G$ is constructed as in Definition~\ref{def:luw} for some prime power $q$ and some $t \in (0,2], s\in [0,1)$ (such  that $q^t,q^s$ are integers). Each vertex $(\ell) = (\ell_1,\ell_2,\ell_3) \in \mathcal{V}_1$ corresponds to a storage node and each vertex $(p) = (p_1,p_2,p_3) \in \mathcal{V}_2$ corresponds to a data block.

Applying item~\ref{item:luwedges} in Definition~\ref{def:luw}, it can be shown that any two distinct nodes, $[\ell^{(1)}] = (\ell^{(1)}_1,\ell^{(1)}_2,\ell^{(1)}_3)$ and $[\ell^{(2)}]= (\ell^{(2)}_1,\ell^{(2)}_2,\ell^{(3)}_3)$, such that $\ell_1^{(1)} = \ell_1^{(2)}$ cannot have a common neighbor.  In particular, if $([\ell^{(1)}],(p))\in E$ and $([\ell^{(2)}],(p))\in E$ for any $(p) = (p_1,p_2,p_3)$, then $\ell^{(1)}_1 = \ell^{(2)}_1$ implies that  $(\ell^{(1)}) = (\ell^{(2)})$.  Indeed, 
\begin{align*}
\ell^{(1)}_1 = \ell^{(2)}_1 &\implies p_1\ell^{(1)}_1 = p_1\ell^{(2)}_1 \\
&\implies p_2 - \ell^{(1)}_2 =p_2-\ell^{(2)}_2 \\
&\implies \ell^{(1)}_2 = \ell^{(2)}_2.
\end{align*}
Then using the other requirement for an edge:
\begin{align*}
\ell^{(1)}_2 = \ell^{(2)}_2 &\implies \overline{{p_1}}\ell^{(1)}_2 + p_1\overline{{\ell^{(1)}_2}} = \overline{{p_1}}\ell^{(2)}_2 + p_1\overline{{\ell^{(2)}_2}} \\
&\implies p_3-\ell^{(1)}_3 =p_3-\ell^{(2)}_3 \\
&\implies \ell^{(1)}_3 = \ell^{(2)}_3.
\end{align*}

We now use the fact that no nodes with the same first value in the tuple share a block, along with the fact that there are no 4-cycles, to lower bound the size of the neighborhood of any set of $k$ vertices in $\mathcal{V}_1$.  Recalling that $\mathcal{V}_1$ corresponds to storage nodes, this lower bound is also a lower bound on the storage capacity of the LBFR with $G$ as the adjacency graph.

There are $k \alpha$ edges adjacent to any set of $k$ vertices of $\mathcal{V}_1$, since each vertex in  $\mathcal{V}_1$ has a degree of $\alpha$ ($=q^s$). As shown in the proof of Theorem~\ref{theorem:good}, since there are no 4-cycles any two nodes share at most one neighbor and thus any set of $k$ nodes must have at least $M(k) \geq k\alpha - {k\choose 2}$ distinct neighbors.

This lower bound can next be improved using the fact that no two nodes in the same group $B_\ell^i$ for some $i\in [q^t]$ share a neighbor at all, as shown above. 

Let $S \subset \mathcal{V}_1$ be any set of $k$ nodes.  Let $S_i = S \cap B_\ell^i$ for $i = 0, \ldots, q^t - 1$.  Recall that $q^t = \rho$, so there are $\rho$ such sets $S_i$.  By the above, the are no shared neighbors between any two vertices $[\ell],[\ell'] \in S_i$ for any $i$.  Therefore, the number of pairs of points with no shared neighbors is at least
\[ \sum_{i=0}^{q^t - 1} { |S_i| \choose 2 } \geq (\rho - r){b \choose 2} + r { b+1 \choose 2} = \rho {b \choose 2} + br, \]
where the inequality follows from the fact that the quantity on the left is minimized when the $|S_i|$ are as balanced as possible.  There are $\rho$ sets $S_i$ and $|S| = k = \rho b + r$, so this occurs when there are $\rho - r$ sets of size $b$ and $r$ sets of size $b+1$.

Further, because $G$ has no four-cycles, any pair of points $[\ell],[\ell'] \in S$ have at most one shared neighbor.  This implies that the total number of neighbors of $S$ is at least
\[ k \alpha - {k \choose 2} + \rho {b \choose 2} + br. \]

\end{proof}


This lower bound is the same capacity $M(k)$ for the construction of FR codes via transversal designs in \cite{silberstein2015optimal}. 
As shown in~\cite{silberstein2015optimal},  for $\rho \geq 3$, $k = b\rho+r$, $k\leq \alpha$,  $0 \leq r \leq \rho -1$, and $\alpha > \alpha_0(k)$, the storage capacity $M(k) = k\alpha - {k \choose 2} + \rho {b \choose 2} + br$ achieves the optimal capacity, $g(k)$ for all  $k\leq \alpha$ (defined in Section~\ref{sec:relatedwork}) and is thus optimal, where 
$$\alpha_0(k)  = \begin{cases}
      \frac{b^2\rho{\rho-1\choose 2} + (\rho-2){r-1\choose  2}+b((\rho^2+1)(r-1)-\rho(3r-4))}{\rho-r+1}  \text{ if }k \not \equiv 0 \pmod{\rho}
        \\
       b(b\rho-2){\rho-1 \choose 2}+b-1 \text{ if }k  \equiv 0 \pmod{\rho}.
        \end{cases}
        $$

 Hence our construction of LBFR codes from LUW graphs is also optimal when these constraints are satisfied; in particular this is possible when $\rho=3$ and $\alpha \leq 5$.



\section{Conclusion}
We have described a strengthening of Fractional Repetition (FR) codes, called Load-Balanced Fractional Repetition (LBFR) codes, which have an additional load-balancing property compared to general FR codes when repairing more than one failed node. We show that LBFR codes are precisely characterized as FR  codes with adjacency graphs of girth eight. Moreover, we show that LBFR codes are automatically universally good, meaning they have at least the storage capacity of MBR codes. To motivate the definition of LBFR codes, we show how they lead to faster repair in the presence of sub-packetization under a natural communication model. In particular, in the model considered, where each block is made up of $T$ packets, repair takes time $T(1+o(1))$ timesteps. In contrast, a general FR code can take longer than that (a larger multiple of $T$, which depends on the expansion of the underlying graph). Finally, we describe explicit constructions of LBFR codes, including examples which achieve the optimal storage capacity for FR codes.

\bibliographystyle{plain}
\bibliography{citations}

\end{document}